%% file: Main.tex
\documentclass[conference]{IEEEtran}
\IEEEoverridecommandlockouts
\usepackage{cite}
\usepackage{amsmath,amssymb,amsfonts}
\usepackage{graphicx}
\usepackage{textcomp}
\usepackage{xcolor}
\usepackage{multirow}

\usepackage{bm}
\usepackage{enumerate}
\usepackage{threeparttable}
\usepackage{tabularx} 
\usepackage{float}
\usepackage{caption}
\usepackage{mathtools}
\usepackage[shortlabels]{enumitem}
\usepackage{pgfgantt}
\usepackage{tablefootnote}

\usepackage{graphicx}
\usepackage{placeins}
\usepackage[normalem]{ulem}
\usepackage{latexsym}
\usepackage{caption}
\usepackage{subcaption}
\usepackage{amsmath}
\usepackage{amssymb}
\usepackage{wrapfig}
\usepackage{tabularx}
\usepackage{lscape}
\usepackage{float}
\usepackage{xcolor,colortbl}
\usepackage{times}
\usepackage{stfloats}
\usepackage{multirow}
\usepackage{enumerate}
\usepackage{xspace}

\usepackage{calc}
\usepackage{fancyhdr}
\usepackage[shortlabels]{enumitem}
\usepackage{graphicx}
\usepackage{varioref}
\usepackage{lipsum}
\usepackage{url}
\usepackage{tikz}
\usepackage{booktabs}
\usepackage{algorithm}
\usepackage{amsthm}
\usepackage{algcompatible}
\usepackage[noend]{algpseudocode}

\usepackage{multirow}
\usepackage{longtable}
\usepackage{graphicx}
\usepackage{verbatim}

\newcommand{\algrule}[1][.2pt]{\par\vskip.5\baselineskip\hrule height #1\par\vskip.5\baselineskip}
{\bfseries}{\rmfamily}
\newtheorem{mycorollary}{Corollary}{\bfseries}{\rmfamily}
\newtheorem{mylemma}{Lemma}{\bfseries}{\rmfamily}

\usepackage{cite}
\usepackage{textcomp}
\usepackage{bm}
\usepackage{threeparttable}

\newcommand{\as}{\ensuremath {\leftarrow}{\xspace}}

\newcommand{\PQTor}{\textit{PQ-Tor}}

\newcommand{\puzgen}{\ensuremath {\texttt{Puzzle.Gen}{\xspace}}}
\newcommand{\pow}{\ensuremath {\texttt{Puzzle.PoW}{\xspace}}}
\newcommand{\powverif}{\ensuremath {\texttt{PoW.Verify}{\xspace}}}
\newcommand{\createtoken}{\ensuremath {\texttt{Create.Token}{\xspace}}}
\newcommand{\query}{\ensuremath {\texttt{Query}{\xspace}}}
\newcommand{\response}{\ensuremath {\texttt{Response}{\xspace}}}
\newcommand{\client}{\ensuremath {\texttt{Client}{\xspace}}}
\newcommand{\PSB}{\ensuremath {\texttt{PSB}{\xspace}}}

\newcommand{\blockrec}{\ensuremath {\texttt{Block.Reconstruction}{\xspace}}}

\newcommand{\kyberdecap}{\ensuremath {\texttt{Kyber.Decap}{\xspace}}}
\newcommand{\kyberencap}{\ensuremath {\texttt{Kyber.Encap}{\xspace}}}
\newcommand{\kyberkeygen}{\ensuremath {\texttt{Kyber.KeyGen}{\xspace}}}
\newcommand{\dilithiumverify}{\ensuremath {\texttt{Dilith.Verify}{\xspace}}}
\newcommand{\dilithiumsign}{\ensuremath {\texttt{Dilith.Sign}{\xspace}}}
\newcommand{\dilithiumkeygen}{\ensuremath {\texttt{Dilith.KeyGen}{\xspace}}}

\newcommand{\pacdosq}{\ensuremath {\texttt{PACDoSQ}{\xspace}}}

\newcommand{\dbsetup}{\ensuremath {\texttt{DB.Setup}{\xspace}}}

\def\BibTeX{{\rm B\kern-.05em{\sc i\kern-.025em b}\kern-.08em
    T\kern-.1667em\lower.7ex\hbox{E}\kern-.125emX}}
\begin{document}

\title{Privacy-Preserving and Post-Quantum Counter Denial of Service Framework for Wireless Networks\\
}

\author{\IEEEauthorblockN{Saleh Darzi}
\IEEEauthorblockA{\textit{University of South Florida}, Tampa, USA \\
salehdarzi@usf.edu}
\and
\IEEEauthorblockN{Attila Altay Yavuz}
\IEEEauthorblockA{\textit{University of South Florida}, Tampa, USA \\
attilaayavuz@usf.edu}
}

\maketitle
\input{Abstract.tex}


\input{Introduction.tex}

\input{Prelim.tex}

\input{Scheme.tex}

\input{Security.tex}

\input{PerfEval.tex}

\input{Conclusion.tex}


\bibliographystyle{ieeetr} 
\bibliography{SalehRef.bib}

\end{document}

%% file: Abstract.tex
\begin{abstract}
As network services progress and mobile and IoT environments expand, numerous security concerns have surfaced for spectrum access systems (SASs). The omnipresent risk of Denial-of-Service (DoS) attacks and raising concerns about user privacy (e.g., location privacy,  anonymity) are among such cyber threats. These security and privacy risks increase due to the threat of quantum computers that can compromise long-term security by circumventing conventional cryptosystems and increasing the cost of countermeasures. While some defense mechanisms exist against these threats in isolation, there is a significant gap in the state of the art on a holistic solution against DoS attacks with privacy and anonymity for spectrum management systems, especially when post-quantum (PQ) security is in mind. 
In this paper, we propose a new cybersecurity framework \pacdosq, which is the first to offer location privacy and anonymity for spectrum management with counter DoS and PQ security simultaneously. Our solution introduces the private spectrum bastion concept to exploit existing architectural features of SASs and then synergizes them with multi-server private information retrieval and PQ-secure Tor to guarantee a location-private and anonymous acquisition of spectrum information together with hash-based client-server puzzles for counter DoS. We prove that \pacdosq~achieves its security objectives, and show its feasibility via a comprehensive performance evaluation.  

\end{abstract}

%% file: Introduction.tex
\section{Introduction} \label{sec:Intro}
With the progression of wireless network services such as 5G/6G, coupled with the rapid expansion of mobile and IoT applications, the significance and frequency of threats to such services, specifically spectrum access systems (SAS) are escalating \cite{grissa2021anonymous}. Among these threats, the omnipresence of Denial of Service (DoS) attacks is becoming increasingly sophisticated and executable, due to the availability of open-source software, enhanced processing capabilities, and the proliferation of inexpensive devices. DoS attacks are particularly applicable to emerging wireless networked systems because of their inherent broadcast nature, spectrum access, and geolocation database requisites \cite{chakraborty2023capow}.  Wireless spectrum access, despite its merits, also brings profound privacy concerns for its users. More specifically, the continuous reporting of spectrum and location data to geo-location database servers raises numerous privacy concerns \cite{jasim2021cognitive}. Finally, the emergence of quantum computers poses a significant risk to the long-term security and privacy preservation of these next-generation networks, challenging existing classical security countermeasures \cite{darzi2023envisioning}.  
Efforts are underway to address security issues including counter-DoS, privacy, and PQ threats in SAS. However, existing solutions work in isolation, and do not comprehensively and effectively tackle these issues simultaneously. We outline some of the most relevant efforts to our work below. 


\vspace{-2mm}
\subsection{Related Work} \label{subsec:literature} 
{\em Counter-DoS and Spectrum Management for NextG Networks}: 
The widespread growth of mobile and IoT devices has led to a shortage of spectrum resources.  Cognitive Radio Networks (CRN) offer secondary users (SUs) the ability to opportunistically access unoccupied licensed channels, presenting a prospective solution for spectrum management. While spectrum management serves as a critical wireless resource allocation tool, it faces several security threats, including DoS attacks, due to their broadcast nature, database-driven architecture, and the potential malicious behavior of SUs~\cite{chakraborty2022bankrupting, grissa2015lpos}. Adversaries may target system availability through DoS aiming to exhaust server resources that handle servicing requests. Numerous research endeavors exist to mitigate DoS attacks through intrusion detection systems (IDSs) and mechanisms encompassing network-based solutions, cryptographic techniques, and game theory-based approaches. 
Recent progress in machine learning has propelled AI-based IDSs into the spotlight, demonstrating the ability to accurately identify abnormal behavior, with success rates surpassing $95\%$ in some cases \cite{chakraborty2023capow}. However, despite their merits, these methods may require knowledge and access to broad (some cases private) network topologies, user-sensitive network traffic, and continuous training on large-scale data \cite{darzi2024counter, masur2022artificial}. Moreover, they may be vulnerable to some AI-based loopholes exploited by attackers with substantial costs to the underlying system \cite{doriguzzi2024flad, mittal2023deep}. Therefore, it is ideal that they are complemented with counter-DoS techniques that do not rely on such features and can offer additional provable security guarantees.  


Client Puzzle Protocols (CPP) permit a client to access server resources only upon presenting a valid token generated by solving a puzzle like Proof of Work (PoW) \cite{bostanov2021client}. CPPs significantly increase the cost of the DoS attack (e.g., computational, memory) depending on the type of puzzles (e.g., timing, AI-based), thereby substantially mitigating their impact. CPPs can offer an ideal complement to AI-based counter-DoS, but they must achieve various properties such as cost asymmetry, efficiency, statelessness, memorylessness, unforgeability, and non-parallelization~\cite {ali2020foundations}.  Given their requisite features and the need for scalability for IoT networks, alleviating the burden of puzzle management from the users and servers is crucial. One such effort is outsourcing puzzle generation and distribution to a trusted entity called ``Bastions" \cite{kiruthika2011new}. However, these approaches often presume the existence of Bastions in applications, presenting just abstract concepts without clarity on which entity realistically assumes the Bastion role, thus missing proper architectural incentives. To benefit from outsourced CCPs, bastions' trust level and architectural duty must be well-justified and integrated into the target application. 

{\em Privacy and Anonymity in SAS under DoS Attacks}: The Federal Communications Commission (FCC) has instructed the utilization of centralized SAS, comprising multiple geolocation spectrum databases to foster dynamic spectrum resource access~\cite{grissa2016efficient}. This facilitates spectrum sharing between governmental entities and commercial operators, with primary and secondary users. FCC mandates that users provide sensitive information, including precise location coordinates (longitude and latitude), desired spectrum channel, usage data, and transmission details, to access spectrum availability \cite{jasim2021cognitive}. This not only gives rise to privacy concerns regarding users' confidential data, and identity but also facilitates the tracing and potential exposure of location privacy (e.g., revealing behavioral patterns, lifestyle choices, etc.) \cite{agarwal2022survey}.  Moreover, the absence of authentication during private spectrum data access, coupled with the reliance on many counter-DoS solutions for authentication, underscores the critical need to prioritize user anonymity and privacy. This gap in spectrum management services calls for solutions that address anonymity, location privacy, and DoS mitigation simultaneously.

{\em Counter-DoS and Privacy for Wireless Networks in the Post-Quantum Era}:  The emergence of quantum computers presents a substantial security risk to NextGen networks, potentially compromising foundational security protocols (e.g., TLS) and undermining critical aspects of SAS such as DoS protection and privacy safeguards (e.g., \cite{grissa2021anonymous}). Furthermore, conventional cryptographic methods used in privacy-preserving techniques, anonymity networks, and counter-DoS solutions rely on cryptographic problems vulnerable to quantum computers. Hence, Post-Quantum Cryptography (PQC) becomes imperative to furnish a robust long-term security solution \cite{darzi2023envisioning}. 

\vspace{-2mm}
\subsection{Our Contribution} \label{subsec:contribution}
{\em We designed a novel framework that culminates various cryptographic techniques to address the complex array of privacy and security challenges stemming from SAS under DoS and quantum computer attacks. Our scheme presents a "Privacy and Anonymity preserving Counter-DoS in the post-Quantum era" (\pacdosq) for spectrum management in next-generation networks.} We summarize some of the desirable properties of \pacdosq~as follows: \\ \vspace{-3mm}

$\bullet$~\noindent{\em \uline{Enabling Outsourced Counter-DoS Services with SAS Architecture Compliance}:} We devised innovative counter-DoS services formed on CPP architecture featuring hash-based puzzles, where puzzle generation and distribution are delegated to database-driven entities, termed "Private Spectrum Bastions" (PSBs) in our work. Integrating Bastion services within SAS geo-location databases offers several advantages, as PSBs can supply quantum-safe puzzles alongside spectrum availability, maintaining architectural feasibility, and enhanced efficiency. Our PSB approach also paves the way for tackling the location privacy problems of spectrum management and outsourced puzzle services with enhanced robustness as described below. \\ \vspace{-2mm}

\vspace{-1mm}
$\bullet$~\noindent{\em \underline{Fault-Tolerant Location Privacy and Anonymity}:} The database-driven SAS architecture, by FFC requirements, brings various privacy issues as discussed in Section \ref{subsec:literature}~\cite{grissa2021anonymous}. Therefore, despite our opportunistic integration of outsourced CCPs with existing SAS architectures, which mitigates DoS attacks, it still requires clients to obtain puzzles and spectrum data from PSBs. We address the privacy concerns as follows: {\em (i)} We harness distributed Private Information Retrieval (PIR) protocols~\cite{goldberg2007improving} that synergize with multi-server PSB architecture \cite{grissa2021anonymous}. The clients fetch spectrum information (by adhering to FFC regulations) and CCPs privately. Moreover, our choice of PIR protocol permits resiliency against network failures and some subsets of non-responding PSB servers.  {\em (ii)} We ensure clients connect to PSBs and perform private retrieval operations through a post-quantum secure version of the Tor network~\cite{tujner2019quantum}, thereby offering anonymous access.  \\ \vspace{-2mm}




\vspace{-2mm}
$\bullet$ \noindent{\em \underline{Post-Quantum Security}:} \pacdosq~offers all the above desirable security and privacy features with a post-quantum guarantee thanks to the reliance on NIST-PQC standards in Tor and information-theoretically secure PIR operations. 





%% file: Prelim.tex
\section{Preliminaries and Building Blocks} \label{sec:Prelim} \vspace{-1mm}

In this section, we outline the notations, cryptographic primitives, and tools employed in our proposed framework.


\textbf{Notations:}  $|x|$ and $\{0,1\}^k$ signify the bit length of a variable and $k$-bit binary value, respectively. $\mathbb{F}$, $GF(2)$, and $\mathbb{Z}$, denote a finite field, Galois Field with modulo $2$, and a set of integers, respectively. $\{x_i\}_{i=1}^{\ell}$ and $\xleftarrow{\$}\mathcal{S}$ denote $(x_1, x_2, ..., x_\ell)$ and random selection from the set $\mathcal{S}$, respectively. The function $h(.)$ denotes a cryptographically secure hash function. $sk$ and $pk$ are secret and public keys, respectively.

\textbf{Private Information Retrieval:} The PIR construction enables a client to retrieve a block of information from a database without revealing the privacy of the retrieved item to the database server(s). We will focus on multi-server PIR since our system model includes multiple spectrum databases. We opt for the fault-tolerant IT-PIR \cite{goldberg2007improving} that offers $\nu$-byzantine robustness, ensuring the reconstruction of the target block even if $\nu$ servers provide incorrect responses.

\textbf{PQ-Secure Primitives:} We use NIST PQC standardized lattice-based schemes for KEM and signature,  $\textit{Kyber}$ \cite{bos2018crystals} and $\textit{Dilithium}$ \cite{ducas2018crystals}, respectively.  
The $\textit{Kyber}$ KEM is formed on the Module-LWE problem and is comprised of three algorithms ($\kyberkeygen, \kyberencap, \kyberdecap$). 
The $\textit{Dilithium}$ signature is also formed on Module-LWE and is comprised of three algorithms: $(sk, pk)\as$ $\dilithiumkeygen(1^\lambda)$; $\sigma\as$ $\dilithiumsign(sk, m)$; and $\{0,1\} \as$ $\dilithiumverify(pk, m, \sigma)$.

\textbf{Hash-based Puzzles:} We use hash-puzzles \cite{aura2000resistant} that are comprised of three functions ($\texttt{Gen}{\xspace}, \texttt{PoW}{\xspace}, \texttt{Verify}{\xspace}$): 
\begin{itemize}[leftmargin=*] 
	\item[-] $\Pi \leftarrow Puzzle.Gen(1^\lambda, \kappa)$: Given the security parameter $\lambda$ and the difficulty level $\kappa$, it selects a random nonce $N \as \{0,1\}^\kappa$ and produces hash-based puzzles $\Pi = (N, \kappa)$.
	\item[-] $\Psi \as\pow(\Pi, \kappa)$: Given a puzzle $\Pi$, it brute forces a nonce $\Psi = N_x$ to obtain a hash value with $\kappa$-bit leading zeros, $0_10_20_3...0_\kappa Y \as h(\Pi, N_x)$, where $Y \in \{0,1\}^{|h|-\kappa}$.
	\item[-] $\{0,1\} \as \powverif(\Pi, \Psi)$: The verifier checks if the first $\kappa$ bits of the hash value of $h(\Pi, N_x)$ are zero.
\end{itemize}

%% file: Scheme.tex
\section{The Proposed System \& Framework: $\pacdosq$}\vspace{-1mm}
\label{sec:systemoverview}

This section delineates our system setup and framework.

    
\subsection{PACDoSQ Architecture and Initial Setup}\label{subsec:initialsetup}
Our system model has three main entities: 1) {\em Private Spectrum Bastions (PSBs):} consist of multiple geo-location spectrum databases \cite{agarwal2022survey, grissa2021anonymous} that provide spectrum availability information. They maintain synchronicity and consistency under FCC guidelines. 2) {\em Clients:} are secondary users equipped with mobile devices (e.g., laptops). They connect to the servers for network services by obtaining spectrum availability from PSBs. 3) {\em Servers:} are various network servicing platforms (e.g.,  web/cloud servers), to which clients seek to connect. 





We outline the initial setup of PSB and \textit{PQ-Tor} below.  \vspace{1mm}


\textbf{Database (DB) Structure \& Setup:} The PSBs synchronize their DB by incorporating various parameters like location coordinates ($l_x, l_y$), frequency channel number ($ch$), and spectrum data. DB is conceptualized (and simplified) as a matrix with dimensions $r\times s$, where each row represents one data block comprising $b$ bits. Each block consists of $s$ words, each with a size of $w$ bits formatted as $GF(2^w)$ (as in \cite{chakraborty2022bankrupting, chakraborty2023capow}). PSBs maintain other relevant information stipulated by the FCC (as in \cite{agarwal2022survey}) like row index of coordinates with proper subroutines, for brevity herein referred as  $\textit{DB-Index(.)}$.

\textbf{\textit{PQ-Tor} Configuration:} Our $\PQTor$ variant has the following alterations over conventional Tor: (i) $\textit{RSA}$ signature is replaced with $\textit{Dilithium}$ signing in the consensus part. (ii) $\textit{RSA}$ KEM is substituted with $\textit{Kyber}$ KEM in circuit creation. (iii) $\textit{AES-128}$ is replaced with $\textit{AES-256}$ to double symmetric key size against Grover's algorithm \cite{glas2012signal}.

\subsection{\pacdosq~Framework} \label{subsec:mainoperations}

We illustrate the flow of \pacdosq~framework in Fig. \ref{fig:Scheme}, provide its algorithmic description in Algorithm \ref{Alg:PQInstantiation}, and further elaborate on its steps as follows: 



\begin{figure*}
  \includegraphics[scale=0.51]{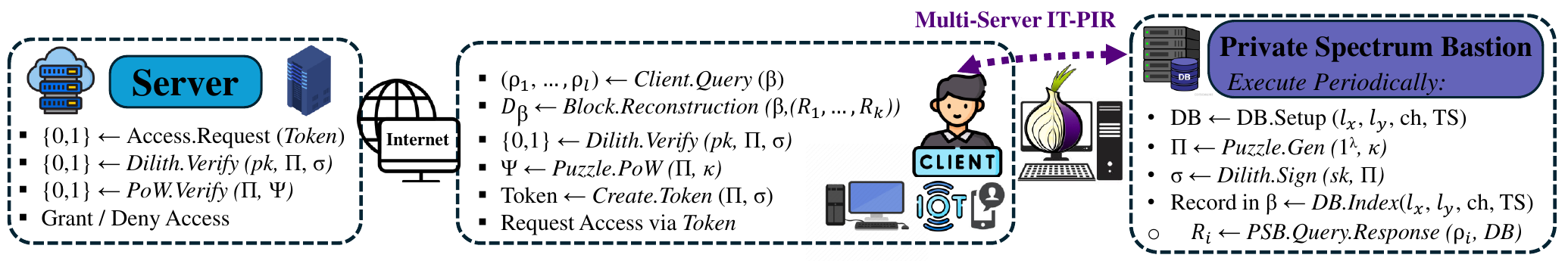}\vspace{-2mm}
  \caption{A high-level representation of the proposed architecture and workflow.}
  \label{fig:Scheme}\vspace{-5mm}
\end{figure*}

\begin{algorithm}[ht!]
	\small
	\caption{$\pacdosq$ Scheme}\label{Alg:PQInstantiation}
	\hspace{5pt}
	\begin{algorithmic}[1]
		\Statex \vspace{-1mm}$\textbf{Private Spectrum Bastions}$:        
            \State $\textit{DB} \leftarrow \dbsetup(l_x, l_y, ch, TS)$
            \For{$(l_x, l_y) \in \textit{grid} $}
            \For{$TS \in \textit{Timeframe}$}
            \For{$\theta \in \{1, 2, ..., \textsf{max}\}$}
		\State $\Pi_\theta \as \puzgen(1^\lambda, \kappa)$ 
            \State $\sigma_\Pi \as \dilithiumsign(sk, \Pi_\theta)$
            \State Given $\beta \leftarrow \textit{DB-Index}(l_x, l_y, ch, TS)$.
		\State Record $\Pi_\theta$ and $\sigma_{\Pi_\theta}$ in \textit{DB} within index $\beta$.
            \EndFor
            \EndFor
            \EndFor
		\Statex $\underline{R_i \as \PSB.\query.\response(\rho_i, \textit{DB})}$: \vspace{+1mm}
            \State Upon receiving request $\rho_i$, $\PSB_i$ computes: $R_i \leftarrow \rho_i \cdot \textit{DB}$
            \State \Return $R_i$, and respond to the client via $\PQTor$
	\end{algorithmic}
	\algrule
	\begin{algorithmic}[1]
		\Statex $\textbf{Clients}$: 
		\Statex $\underline{(\rho_1, \rho_2, ..., \rho_\ell)\as \client.\query(\beta)}$:\vspace{+2mm}
            \State Obtain index $\beta \leftarrow \textit{DB-Index}(l_x, l_y, ch, TS)$
		\State Set $e_{\beta} \leftarrow \overrightarrow{1_{\beta_r}} \in Z_r$
		\State Choose $\ell$ distinct $\alpha_1, \alpha_2, ..., \alpha_\ell \in \mathbb{F}^*$
		\State Choose $r$ random degree-$t$ polynomials $f_1, f_2, ..., f_r \xleftarrow{\$} \mathbb{F}[x]$ s.t. $f_j(0) = e_\beta[j]$, for all $j \in [1, 2, ..., \ell]$
		\State $\rho_i \as \langle f_1(\alpha_j), f_2(\alpha_j), ..., f_r(\alpha_j)\rangle$ for all $j \in \{1, 2, ..., \ell\}$
		\State Query PSB$_i$ via transmitting $\rho_i$ over \textit{PQ-Tor} for $i = 1, 2, ...,\ell$\vspace{+1mm}
		\Statex $\underline{D_{\beta c}\as \blockrec(\beta, (R_1, R_2, ..., R_k))}$:\vspace{+1mm}
		\State Receiving PIR responses $R_i$ from PSBs for $i = 1, 2, ...,k$
		\If{$k>t$}
		\EndIf
		\For{$c$ from $1$ to $s$}
		\State $R_{ic} \leftarrow R_i[c]$ for all $i \in [1, 2, ..., k]$
		\State $S_c \leftarrow \langle R_{1c}, R_{2c}, ..., R_{kc} \rangle$
		\State $D_{\beta c} \leftarrow \textit{EASYRECOVER}(t, \omega, [\alpha_1, \alpha_2, ..., \alpha_k], S_c)$
		\If{Recovery fails and $\nu < k - \lfloor{\sqrt{kt}\rfloor}$} 
		\State $S_c \leftarrow \langle R_{1c}, R_{2c}, ..., R_{kc} \rangle$
		\State $\Return~D_{\beta c} \leftarrow \textit{HARDRECOVER}(t, \omega, [\alpha_1, \alpha_2, ..., \alpha_k], S_c)$
		\EndIf
		\EndFor
		\Statex $\underline{\textit{Token}\as \createtoken(\Pi, \sigma_{\Pi})}$: Given $\Pi$ in $D_{\beta}$, \textbf{do} 
            \State $\{0,1\} \as \dilithiumverify(pk,\Pi, \sigma_\Pi)$
            \State $\Psi \as \pow(\Pi, \kappa)$
            \State $\Return~\textit{Token} \leftarrow (\Pi, \sigma_\Pi, \Psi)$
    	\end{algorithmic}
	\algrule
	\begin{algorithmic}[1]
		\Statex $\textbf{Servers}$: For requests, server does:
            \State $\{0,1\} \as \dilithiumverify(pk, \sigma_{\Pi})$
            \State $\{0,1\} \as \powverif(\Pi, \Psi)$
            \State \textbf{if} above holds, $\Return~1$, grant access, and record the $\textit{Token}$
            \State \textbf{otherwise}, $\Return~0$ and deny access.
	\end{algorithmic}
\end{algorithm} 
\setlength{\textfloatsep}{0pt}

\subsubsection{\textbf{PSBs - Puzzle Management and Private Spectrum Service}} {\em (i)} PSBs setup DB by generating spectrum management context (e.g., coordinates, channels), puzzles, and their PQ signatures as in Step 1-8. Within defined segments of the grid marked by specific coordinates for multiple time frames, they generate hash-based puzzles and sign them according to predetermined indices derived from $((l_x, l_y), ch, TS)$. The puzzles and their $\textit{Dilithium}$ signatures are updated periodically according to the puzzle difficulty/validity interval (e.g., every hour). The quantity of puzzles generated depends on factors such as the number of servers and their maximum capacity ($\textsf{max}$). {\em (ii)} PSBs handle the spectrum query first via the fault-tolerant multi-server PIR~\cite{goldberg2007improving} that permits an information-theoretically private retrieval of coordinate availability, puzzle, and their signatures (Step 9). The PIR response is sent to the client via \PQTor~to ensure anonymity (step 10).

\subsubsection{\textbf{Client's Private Availability Information and Quantum-Safe Puzzle Retrieval}} To comply with FCC regulations and participate in the counter-DoS mechanism for accessing a networking services server, the clients retrieve puzzles and spectrum information from the PSBs. Clients use their coordinates, frequency channel, and timestamp to determine the target index $\beta$ within the PSB's DB (Step 1). Subsequently, the client constructs a PIR request by selecting a basis vector $\overrightarrow{1_{\beta_r}}$, where all elements are zero except for index $\beta$, which is set to one (Step 2). Furthermore, considering $\ell$ PSBs and utilizing Shamir's secret sharing technique, the client selects $\ell$ random elements from $\mathbb{F}^*$ (Step 3), generates $r$ random polynomials with a degree of $t$ satisfying $f_j(0) = e_\beta[j]$ (Step 4), and creates $\ell$ PIR requests $\rho$ (Step 5). Finally, the client dispatches the PIR requests to each PSB's $DB_i$ via $\PQTor$ (Step 6).

Steps 7-15 involve the client's query recovery phase. Assuming that $k$ out of $\ell$ PSB servers respond to the client, the client can reconstruct the block using the $\textit{EASYRECOVER}$ subroutine as described in \cite{goldberg2007improving}, which relies on the Lagrange interpolation technique. If a sync/transmission error occurs or an incorrect block is returned by $\nu<k$ servers (e.g., Byzantine (compromised) server), the client can use $\textit{HARDRECOVER}$ algorithm~\cite{goldberg2007improving} based on error-correction codes to handle the error.   By reconstructing the block item with one of these recovery algorithms, the client retrieves the puzzle, whose validity can be confirmed by verifying PSB's signature.

\subsubsection{\textbf{PoW \& Token Creation}}
The online phase of the framework begins at this stage, where the client performs the PoW and generates the $\textit{Token}$. 
Given the hash-based puzzle ($\Pi$) and the target network service ID ($ID_S$), the client must conduct a brute-force search through a nonce ($N_C$) to discover a hash value $h(ID_S, TS, N_B, N_C)$ with $\kappa$-bit leading zeroes (Step 17-18). Then, upon identifying a solution, it generates the $\textit{Token}$, which comprises the PSBs' and client's nonces along with the $TS$ and $ID_S$, and transmits it to the server.

\subsubsection{\textbf{Access Requests}} 
The client submits a request to the server with $\textit{Token}$ for a given time interval. The server first verifies the puzzle's validity by checking the PSB's signature (Step 1), followed by efficiently verifying the $\textit{Token}$ using a hash operation (Step 2). Only if the puzzle solution is valid and authentic, the access is granted. 

%% file: Security.tex
\section{Security Analysis}
\label{sec:Secuirty}

\textbf{Threat Model and Security Objectives:}  
Our threat model captures a vast range of attacks at the intersection of counter-DoS, privacy, anonymity, and basic security services, all under quantum computing threat: {\em (i)} Clients may launch DoS attacks on the servers. To mitigate such attacks, we consider the counter DoS threat model in outsourced puzzle settings, wherein PSBs carry over the Bastion role for puzzle management. {\em (ii)} Client's location privacy and identity information are under threat due to the FCC's requirement of sharing coordinate and device specs with spectrum management databases. In our model, PSBs carry out this duty along with puzzle management. Hence, we consider that PSBs are curious about the location and identity information of the clients. {\em (iii)} Some (but small set of) PSBs might be compromised and therefore may act as Byzantine servers (do not respond or provide incorrect input). {\em (iv)} The attacker is quantum computing capable and can use it to launch attacks considered in {\em (v)} as well as to threaten basic security services such as confidentiality, authentication, and integrity (which are usually achieved through essential services like TLS). Given the threat model, $\pacdosq$ aims to achieve the following security objectives: 

\begin{itemize}[leftmargin=*]
   \item \underline{\textit{Client Privacy and Anonymity:}} Clients' location privacy (i.e., coordinates), device specs, and identity remain confidential and anonymous during spectrum availability and puzzle retrieval from the PSBs and external attackers. 
    \item \underline{\textit{Resilience to Partial Failure and Byzantine Behavior:}} The client can retrieve and reconstruct the intended block item (including spectrum and puzzle data) even if some subset of the PSBs act non-responsive or malicious. 
    \item \underline{\textit{DoS Mitigation:}} A measurable and provable counter DoS measures are employed. 
    \item \underline{\textit{PQ-security:}} All the above objectives are achieved in the presence of quantum computing capable adversaries. 
\end{itemize}

\vspace{1mm}
\textbf{Security Analysis:} We give a series of security proofs capturing the threat model as follows:


\vspace{-1mm}
\begin{mylemma}
\label{lem:privacy}
$\pacdosq$ ensures $\textit{t-private k-out-of-}\ell$ information-theoretically secure location privacy and computationally secure anonymity via onion routing. 
\end{mylemma} 
\begin{proof}\vspace{-1.5mm}
By utilizing $(\ell,t)$-Shamir secret sharing, with the assumption of $k$ honest responses from $\ell$ PSBs where $k>t$, the target index $\beta$, along with the client's private information, including location and transmission details, remains confidential during the block retrieval process, even in the event of collusion among $t$ PSBs with ${0\leq t\leq\ell-1}$. The deployment of onion routing with a minimum of three intermediate nodes, each possessing knowledge solely of its predecessor and successor, alongside communication through a circuit with layers of symmetric encryption using $\textit{AES-256}$ keys derived via a $\textit{Module-LWE}$-based KEM scheme, ensures the anonymity and untraceability of the client's identity and activities against both PSBs and eavesdropping adversaries.
\end{proof}

\vspace{-1mm}
\begin{mycorollary}\label{cor:Robustness}
$\pacdosq$ attains $\nu$-Byzantine-Robustness with $\nu < k - \lfloor \sqrt{kt}\rfloor$. 
\end{mycorollary}
\begin{proof}\vspace{-1.5mm}
\pacdosq~offers block reconstruction from client-received query responses (e.g., communication failures, malicious drop)  
 by employing Guruswami-Sudan list decoding algorithm capable of correcting ${\nu < k - \lfloor \sqrt{kt}\rfloor}$ errors and  $(\ell,t)$-Shamir secret sharing with $k$ responding PSBs ($k>t$). 
\end{proof}




\vspace{-1mm}
\begin{mycorollary}\label{lem:PQDoS}
\pacdosq~offers enhanced counter-DoS for the servers via client-server puzzles.
\end{mycorollary}
\begin{proof}\vspace{-1.5mm}
The server only accepts puzzle solutions with PSB's signature, eliminating the possibility of puzzle forgery. Since PoW requires $O(2^{n})$ trial (for classical settings), the adversary needs an average of $O(2^{\kappa})$ hash operations to acquire a valid $\textit{token}$ for server, where a puzzle is only valid for a designed amount of time depending on the difficult level $\kappa$.  
\end{proof}



\vspace{-2mm}
\begin{mylemma}\label{cor:PQsecurity}
$\pacdosq$ achieves the objectives in Lemma 1 and Corollary 1-2 with PQ-security. 
\end{mylemma}
\begin{proof}\vspace{-1.5mm}
{\em (i)} The location privacy guarantees and robustness features in Lemma \ref{lem:privacy} and Corollary \ref{cor:Robustness}, respectively, are information-theoretically secure, and therefore remain unaffected by the adversary's computational power, including quantum computers \cite{goldberg2007improving}. {\em (ii)} The onion routing anonymity in Lemma \ref{lem:privacy} relies on $128$-bit PQ security of the $\textit{AES-256}$ \cite{bonnetain2019quantum} given Grover's algorithm and the hardness of the $\textit{Module-LWE}$ problems, which can be closely reduced from the worst-case $\textit{Module-SIVP}$ problem in the random oracle model \cite{bos2018crystals}. {\em (iii)} The end-to-end security and PQ-TLS security of PQ-Tor and authentication of puzzles also achieve the same level of PQ-security via NIST PQC framework\cite{darzi2023envisioning}. {\em (iv)} The hash-based puzzle in Corollary 2 offers $O(2^{\kappa/2})$ level of PQ security due to Grover's probabilistic algorithm, and by adjusting the time validity of the PoW accordingly, the hash-based puzzles offer robust PQ counter-DoS mitigation for \pacdosq.   \end{proof}


%% file: PerfEval.tex
\section{Performance Evaluation} \label{sec:PerfEval}\vspace{-1mm}
We present our evaluation metrics and experimental results. \vspace{-3mm}

\vspace{-2mm}
\subsection{Metrics, Selection Rationale, and Configurations} \label{subsec:Configuration}
\noindent \textbf{Evaluation Metrics and Rationale:} We consider the computational, communication, and storage overhead of   \pacdosq~for multi-server PIR, puzzle generation, PoW, token verification, and overhead of PQC, including $\PQTor$ components.  We also investigate scalability aspects such as end-to-end delay perceived by the client for an increased number of users, networking conditions, and PSB configurations. The configuration of PSB servers specifies the privacy levels achieved during block retrieval; for instance, $(3,2)$ indicates that privacy is maintained if any 2 out of 3 PSBs collude. To the best of our knowledge, \pacdosq~is the first to offer location privacy, anonymity, and resiliency for puzzle-based counter DoS with PQ-security. Therefore, a vis-a-vis performance comparison with counterparts is not feasible. Instead, we focus on providing a detailed performance evaluation for given metrics to assess the potential feasibility of our framework. We detail our performance evaluation as follows.\\ \vspace{-2.5mm}




\noindent\textbf{Hardware, Software Libraries, and Parameters:}  We used a desktop equipped with an $11^{th}$ Gen Intel Core $\textit{i9-11900K} @ 3.50~GHz$, $64.0~GiB$ RAM, a $1 TB$ SSD, and Ubuntu $22.04.4~LTS$. We employed varying number of Virtual Machines (VMs) with Ubuntu to simulate multiple PSB/PIR/PQ-ToR interactions. We used \textit{percy++} library\footnote{\url{https://percy.sourceforge.net/}} for the multi-server PIR, the \textit{Open Quantum-Safe} library\footnote{\url{https://openquantumsafe.org/}} for PQC primitives, and the \textit{OpenSSL}\footnote{\url{https://www.openssl.org/}} for hash. The PSB used  \textit{SQLite}\footnote{\url{https://www.sqlite.org/}} and \textit{Python3}.  We used $\textit{AES-256}$, $\textit{Kyber}$ for the KEM part, and $\textit{Dilithium}$ for the signature part of $\PQTor$. The hash-based puzzles are formed on $\textit{SHA-256}$. We rely on NIST-PQC level I security for \textit{Kyber} \cite{bos2018crystals} and \textit{Dilithium} \cite{ducas2018crystals}. \\ \vspace{-2.5mm}

\noindent\textbf{Data and Format Selection:} The database structure is a matrix with varying row sizes (e.g., $2^{10}$, $2^{12}$, $2^{14}$, $2^{17}$), where each row represents a single block of data. Utilizing publicly available raw data from the FCC\footnote{\url{https://enterpriseefiling.fcc.gov/dataentry/public/tv/lmsDatabase.html}}, we estimated that each block in the database would contain approximately $560$ bytes of information, excluding puzzles and signatures.  
Within a designated grid segment defined by coordinates $l_x$ and $l_y$, we populated databases with synthetic data representing spectrum information and signed hash-based puzzles stored in PSBs, synchronized as mandated by the FCC \cite{grissa2021anonymous}.



\vspace{-2mm}
\subsection{Experimental Results} \label{subsec:Performance} \vspace{-1mm}
\noindent\textbf{Computational Costs:} We outline our analysis in Table \ref{tab:Computational} and elaborate it as follows:
{\em (i)} The $\textit{Dilithium}$ signature with the puzzle entails key generation, signing, and verification of $29 \mu s$, $84 \mu s$, and $30 \mu s$, respectively. Puzzle generation and verification each require approximately one hash, while solving (PoW) demands brute force corresponding to the difficulty levels denoted by $\kappa$. The difficulty level is  $\kappa/2$ for quantum attacks with Grover's algorithm. 
{\em (ii)} With $t_\oplus$ representing the time for one \textit{XOR}, analytical costs are ${(n/w)\cdot t_\oplus}$ for PIR computations on the client side and ${\ell\cdot(\ell-1)\cdot r\cdot t_\oplus+3\ell\cdot (\ell+1)\cdot t_\oplus}$ on the PSB side. The empirical costs, as detailed in Table \ref{tab:Computational}, show that the expenses for PIR increase linearly with the size of the database.  
{\em (iii)}  \textit{PQ-Tor}'s costs are circuit build and applying encryption layers dominated by the three $\textit{Kyber}$ and $\textit{AES}$ operations. Notably, $\textit{Kyber}$ key generation, encapsulation, and decapsulation each take $10 \mu s$, $13.4 \mu s$, and $9 \mu s$, respectively, while $\textit{AES-256}$ only costs $7~\mu s$ for key generation and $8~\mu s$ for encryption. \\ \vspace{-2.5mm}



\vspace{-1mm}
\begin{table}[hbt!]
    \centering
    \renewcommand{\arraystretch}{1.2} 
    \setlength{\tabcolsep}{2pt}
    \begin{tabular}{|c|c|c|c|c|c|}\hline
        {\textbf{Entity}} & \textbf{Operations} &\multicolumn{4}{c|}{\textbf{Parameter}} \\\hline\hline
        \multirow{3}{*}{\textbf{PSB}} & \multirow{3}{*}{\textit{Puzzle Generation\& Sign}} & \multicolumn{4}{c|}{$|\textbf{DB}|$}\\\cline{3-6}
         &&$\boldsymbol{2^{10}}$ &$\boldsymbol{2^{12}}$ & $\boldsymbol{2^{14}}$ &$\boldsymbol{2^{17}}$\\\cline{3-6}
         & & $31$~ms & $310$~ms & $3.1$~s & $31$~s\\\cline{2-6}
         &{\textit{Query Response}} & $2.3$~ms & $5.4$~ms & $17.3$~ms & $109.9$~ms\\\hline\hline
         \multirow{6}{*}{\textbf{Client}}& \textit{Query} \& \textit{Reconstruction} & $0.9$~ms & $2.1$~ms & $5.7$~ms & $12.5$~ms\\\cline{2-6}
         & {\textit{Puzzle Signature Verify}}& \multicolumn{4}{c|}{$30~\mu s$}\\\cline{2-6}
         & \textit{PQ-Tor Computations}& \multicolumn{4}{c|}{$255.6~\mu s$}\\\cline{2-6}
         & \multirow{2}{*}{Proof of Wok}& $\boldsymbol{\kappa}:14$ & $\boldsymbol{\kappa}:18$ & $\boldsymbol{\kappa}:20$ & $\boldsymbol{\kappa}:23$\\\cline{3-6}
         && {$5.73$~ms} & {$91.7$~ms} & {$367$~ms} & {$2.93$~s}\\\hline\hline
         \multirow{2}{*}{\textbf{Server}}& {\textit{Puzzle Signature Verify}}& \multicolumn{4}{c|}{$30~\mu s$}\\\cline{2-6}
        & \textit{Token Verification} & \multicolumn{4}{c|}{$0.35~\mu s$}\\\cline{1-6}        
    \end{tabular}
    \begin{tablenotes}
       \item One client, one PSB, and a server in a (3,2) configuration setting, fixed block size of $2.93~KB$, and varying database entries ($|\textit{DB}|$). \vspace{-2mm}
     \end{tablenotes}
    \caption{ Computational Costs of $\pacdosq$}
    \label{tab:Computational}\vspace{-3mm}
\end{table}

\noindent\textbf{Communication and Storage Overhead:} We summarized our findings in Table \ref{tab:Communication} and explain them as follows: {\em (i)} Multi-server PIR is the predominant cost due to its communication overhead. The communication cost of retrieving ${\sqrt{nw}}$ bits from $\ell$ PSBs is approximately ${\ell \times \sqrt{nw}}$. The transmitted data volume increases linearly with the number of database entries. 
{\em (ii)} The storage overhead at the client side is minimal, but that of PSBs increases linearly with the number of puzzles and signatures. The hash-based puzzle $\Pi = (\kappa, N_{B})$ features a difficulty level $\kappa$ of 4 bytes and a nonce $N_{B}$ of $\kappa$ bits, with a \textit{Dilithium} signature size of $2.363$ KB. Given these specifications, each block has a fixed size of $2.93$~KB, resulting in database sizes of $4.1$ MB, $16.8$ MB, $67.3$ MB, and $538.2$ MB respectively for a grid segment.  
{\em (iii)} The communication aspect of $\PQTor$ closely mirrors conventional Tor, with negligible differences (e.g., $\textit{Kyber}$ ops). Thus, we utilized conventional Tor network metrics for communication delay estimation~\cite{TorMetrics}. Despite $\textit{Kyber}$ being faster than conventional $\textit{RSA}$ used in Tor, employing $\textit{Kyber}$ necessitates two packet transmissions due to Tor's default packet size of $512$ bytes, resulting in an average bound of $300~ms$ for circuit build time. The communication delay entails the average timing of sending PIR requests and receiving PIR responses via PQ-Tor within a built circuit. 


\begin{table}
    \centering
    \renewcommand{\arraystretch}{1.2} 
    \setlength{\tabcolsep}{2.1pt}
    \begin{tabular}{|c|c|c|c|c|}\hline
        \multirow{2}{*}{\textbf{Operation}} &\multicolumn{4}{c|}{$|\textbf{DB}|$}\\\cline{2-5}
        & $\boldsymbol{2^{10}}$ & $\boldsymbol{2^{12}}$ & $\boldsymbol{2^{14}}$ & $\boldsymbol{2^{17}}$\\\hline\hline
        \textit{Total Communication}& $12.98$~KB & $25.99$~KB & $77.69$~KB & $605.92$~KB\\\hline
        \textit{Client's Storage} & $4.19$~KB & $17.2$~KB &$68.9$~KB &$597.13$~KB \\\hline
        \textit{PSB's Storage} & $4.1$~MB & $16.8$~MB & $67.3$~MB & $538.2$~MB\\\hline
\textit{Communication Delay} &$\approx 145$~ms &$\approx 175$~ms & $\approx 275$~ms & $\approx 650$~ms\\\hline
        \textit{Circuit RTT Latency} & \multicolumn{4}{c|}{$\approx 250$~ms} \\\hline
    \end{tabular}
     \vspace{-1mm}
    \caption{ Communication/Storage Overhead of $\pacdosq$} \vspace{-4mm}
    \label{tab:Communication}
\end{table}

\vspace{+1mm}
\noindent\textbf{Scalability Assessment:} We assess the performance of \pacdosq~for a growing number of clients and different PSB configurations, offering different privacy-speed trade-offs. Our evaluation combines computational and communication overhead to analyze the perceived end-to-end delay for the clients and PSB servers. The client can retrieve its puzzle along with spectrum availability information {\em offline}. The process involves fetching the signed puzzle from PSBs using multi-server PIR over PQ-Tor. The end-to-end delay encompasses the PIR computation on both the client and PSB sides plus the communication delays due to PQ-Tor when fetching a block of data from multiple PSB databases. The experimental analysis of \pacdosq~for numerous clients with various privacy configurations is depicted in Fig. \ref{fig:scalability}. Upon successfully retrieving the puzzle, the client can connect with the server efficiently by solving PoW and sending its solution with \textit{Token}. This (online) phase is swift, and mirrors standard client-server puzzle settings, with the key difference being that the request is transmitted through a PQ-secure TLS channel.




\begin{figure}
  \includegraphics[scale=0.26]{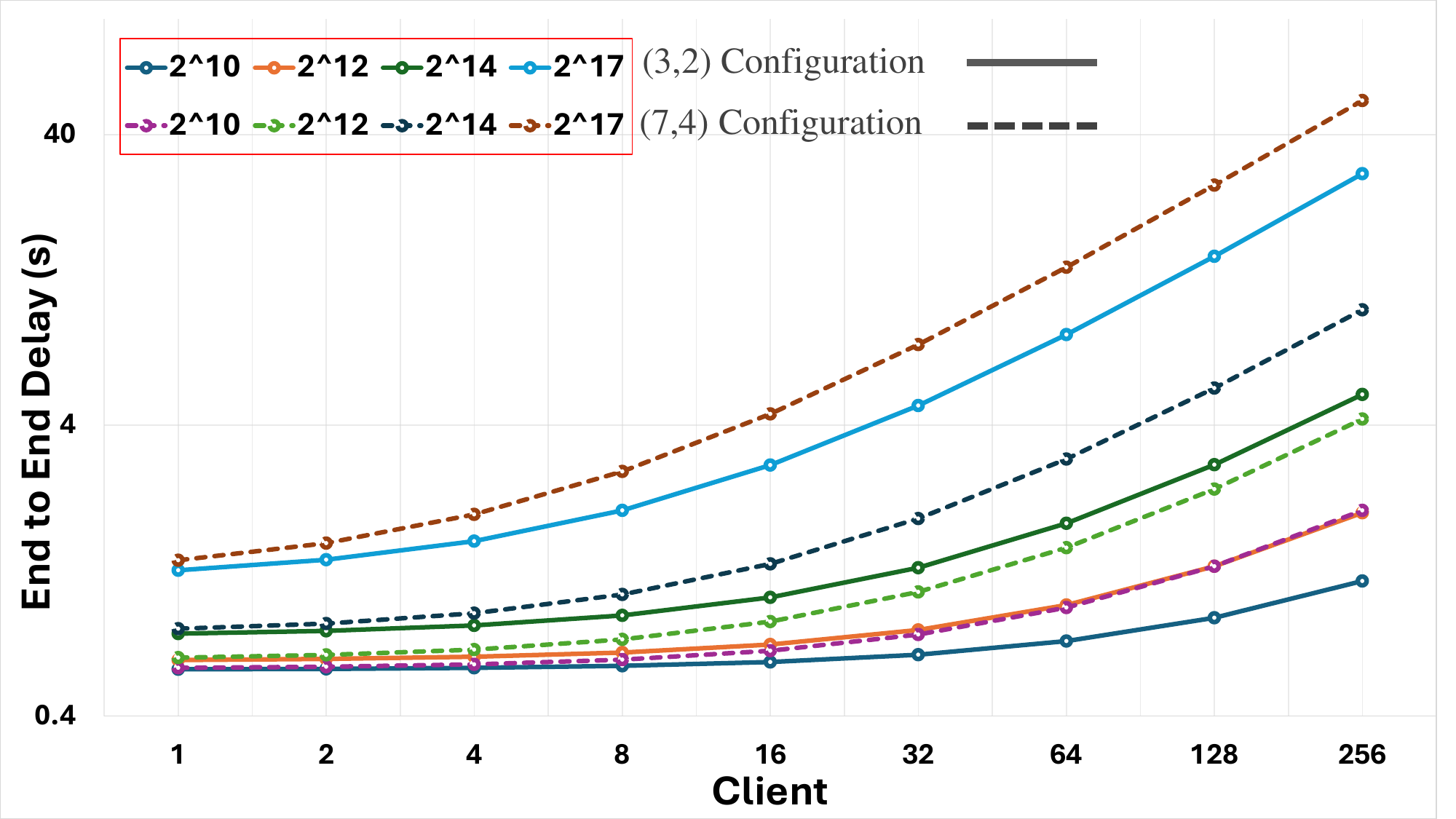}\vspace{-0.5mm}
  \caption{End-to-End delay of $\pacdosq$ for increasing clients.}
  \label{fig:scalability}
\end{figure}

%% file: Conclusion.tex
\section{Conclusion} \label{sec:Conclusion}\vspace{-2mm}
We present \pacdosq, a novel cybersecurity framework designed to address the multifaceted challenges of security, privacy, and DoS attacks in SAS amidst the expanding mobile and IoT landscape and the looming threat of quantum computing. By integrating PSBs with multi-server PIR, PQ-secure Tor, and hash-based client-server puzzles, \pacdosq offers a comprehensive solution that ensures location privacy, anonymity, and resilience against DoS attacks in the PQ era. Formal security proofs validate the security of \pacdosq, while comprehensive performance evaluations underscore its feasibility and efficiency. As network services continue to evolve, \pacdosq~stands as an important step towards establishing a holistic cybersecurity framework safeguarding spectrum management systems from a myriad of cyber threats with reasonable overhead. 

\vspace{-2mm}
\section{Acknowledgment}\vspace{-2mm} This research is partially supported by the NSF CNS-2350213.\vspace{-3mm}